\documentclass {paper}

\usepackage {a4wide}
\usepackage {amssymb}
\usepackage [usenames] {color}
\usepackage {enumerate}
\usepackage {plaatjes}
\usepackage [usenames] {color}
\usepackage [left,pagewise] {lineno}
\usepackage[noend,linesnumbered]{algorithm2e}
\usepackage {amsmath}
\usepackage {amsthm}
\usepackage {xspace}
\usepackage{tabularx}

\definecolor {infocolor} {rgb} {0.6,0.6,0.6}

\definecolor {sepia} {rgb} {0.6,0.2,0.1}

\newcommand {\mathset} [1] {\ensuremath {\mathbb {#1}}}

\newcommand {\R} {\mathset {R}}

\newcommand {\etal} {\textit {et al.}}
\newtheorem {theorem} {Theorem}

\newtheorem {lemma}[theorem] {Lemma}
\newtheorem {observation}[theorem] {Observation}

\newtheorem {claim}[theorem] {Claim}

\theoremstyle {definition}
\newtheorem* {remark} {Remark}

\newcommand{\afunc}{\textsc}

                            {\end{description}}

\title{How Many Potatoes are in a Mesh?}

\author
{
  Marc van Kreveld\thanks
  {
Department of Information and Computing Sciences, Utrecht University, the Netherlands.
{\tt m.j.vankreveld@uu.nl, m.loffler@uu.nl}.
  }
  \and Maarten L\"offler\footnotemark[1]
  \and J\'anos Pach\thanks
  {
  Ecole Polytechnique F\'ederale de Lausanne and R\'enyi Institute, Budapest. 
{\tt pach@cims.nyu.edu}.
 }
}

\begin{document}
\maketitle

\begin{abstract}
We consider the combinatorial question of how many convex polygons
can be made by using the edges taken from a fixed triangulation
of $n$ vertices.
For general triangulations, there can be exponentially many:
we show a construction that has $\Omega(1.5028^n)$ convex polygons,
and prove an $O(1.62^n)$ upper bound in the worst case. 
If the triangulation is fat (every triangle has its angles lower-bounded 
by a constant $\delta>0$), then there can be only
polynomially many: $\Omega(n^{\frac12\lfloor\frac{2\pi}\delta\rfloor})$ and $O(n^{\lceil\frac\pi\delta\rceil})$.

We also consider the problem of counting convex outerplanar polygons 
(i.e., they contain no vertices of the triangulation in their interiors)
in the same triangulations. In this setting, we get 
the same exponential bounds in general triangulations, 
$\Omega(n^{\lfloor\frac{2\pi}{3\delta}\rfloor})$ and
$O(n^{\lfloor\frac{2\pi}{3\delta}\rfloor})$ in fat triangulations.
If the triangulation is furthermore compact (the ratio between the longest and shortest distance between any two vertices is bounded), the bounds drop further to $\Theta (n^2)$ for general convex outerplanar polygons, and $\Theta (n)$ for fat convex outerplanar polygons.
\end{abstract}

\section{Introduction}

It is a common task in combinatorial geometry to
give lower and upper bounds for the number of occurrences of
a certain subconfiguration in a geometric structure.
Well-known examples are the number of vertices in the lower envelope
or single face in an arrangement of line segments, the number of 
triangulations that have a given set of points as their vertices,
etc.~\cite{ps}. In graph drawing, counts of substructures of
graphs are also commonly studied. The crossing number is the best known
example; number of orientations of edges and page numbers are other examples.

In this paper we analyze how many convex polygons (potatoes)
can be constructed by taking unions of triangles taken from a fixed
triangulation (mesh) $M$ with $n$ vertices. Equivalently, we analyze
how many convex polygon boundaries can be made using the edges of 
a fixed triangulation, see Figure~\ref{fig:intro-example}.
For general triangulations there can be exponentially many. 
However, the lower-bound examples use many triangles with very small angles. 
When $n\rightarrow\infty$, the smallest angles tend to zero. To understand
if this is necessary, we also study the number of convex polygons in a triangulation
where all angles are bounded from below by a a fixed constant. 
It turns out that the number of convex polygons is polynomial in this case.

\eenplaatje {intro-example} {A mesh $M$. Three convex polygons that respect $M$ are marked.}

We also study the same questions when the convex polygon cannot have
vertices of $M$ interior to it (carrot). This is the same as requiring 
that the submesh bounded by the convex polygon is outerplanar.

\subsection{Related work}

This paper is motivated by the \emph {potato peeling problem}:
Find a maximum area convex polygon whose vertices
and edges are taken from the triangulation of a given point
set~\cite{akls} or a given polygon~\cite {cy-psppp-86,g-lcpcn-81}.

In computational geometry, \emph{realistic input models} have received considerable
attention in the last decades. By making assumptions on the input, many
computational problems can be solved provably faster than what is possible
without these assumptions. One of the early examples concerned fat triangles:
a triangle is $\delta$-fat if each of its angles is at least $\delta$, for
some fixed constant $\delta>0$. 
Matousek~\etal~\cite{mpssw} show that the union of
$n$ $\delta$-fat triangles has complexity $O(n\log \log n)$ while for $n$ 
general triangles this is $\Omega(n^2)$. As a consequence, the union
of fat triangles can be computed more efficiently as well. 

In~\cite{abt,bchlt,bsvk,mks}, \emph{fat triangulations} were used as a realistic
input model motivated by polyhedral terrains, sometimes with extra assumptions.
Fat triangulations are also related to the meshes computed in the area of
high-quality mesh generation. The smallest angle of the 
elements of the mesh is a common quality measure~\cite{beg}.
In graph drawing, an embedded planar straight-line graph is said to have constant 
\emph{angular resolution} any two edges meeting at a vertex make an angle that
is at least a constant. Hence, fatness and constant angular resolution are the
same for triangulations.


The original definition of realistic terrains applied to triangulations has 
stronger assumptions than fatness~\cite{mks}. Besides fatness of the triangles,
it assumes that any two edges in the triangulation differ in length by at most
a constant factor (compact), and the outer boundary of the triangulation is a fat convex
polygon (bounded aspect ratio). 

\subsection {Results}

We present lower and upper bounds on the maximum number of convex polygons in a mesh in several settings.
The input can be either a \emph {general} mesh, a \emph {fat} mesh (where every angle of each triangle is at least $\delta$), or a \emph {compact fat} mesh (where additionally,
the ratio between the shortest and longest edge is at most $\rho$).
The output can be either a \emph {potato} (general convex submesh) or a \emph {carrot}
(outerplanar convex submesh, that is, one that contains no vertex of the underlying 
mesh in its interior), and each can additionally be required to be \emph {fat} 
(where the ratio between the largest inscribed disk and the smallest containing 
disk is at most $\gamma$).

\begin {table} [t] \label {tab:results}
\centering
\begin{tabular}{|lllll|}
\hline
  input mesh & output vegetable & lower bound & upper bound & source\\
\hline
\hline
  general        & fat carrots  &  $\Omega(1.5028^n)$ && Section~\ref {sec:gen-lo}\\
  general        & anything     && $O(1.62^n)$ & Section~\ref {sec:gen-up}\\
  fat            & fat potatoes &  $\Omega(n^{\frac12\lfloor\frac{2\pi}\delta\rfloor})$ && Section~\ref {sec:fat-lo} \\
  fat            & anything     && $O(n^{\lceil\frac\pi\delta\rceil})$ & Section~\ref {sec:fat-up}\\
  fat            & fat carrots  &  $\Omega(n^{\lfloor\frac{2\pi}{3\delta}\rfloor})$ && Section~\ref {sec:car-lo}\\
  fat            & carrots      && $O(n^{\lfloor\frac{2\pi}{3\delta}\rfloor})$ & Section~\ref {sec:car-up}\\
  compact fat & carrots      &  $\Omega(n^2)$ & $O(n^2)$ & Section~\ref {sec:com-lo} and~\ref {sec:com-up}\\
  compact fat & fat carrots  &  $\Omega(n)$ &  $O(n)$ & Section~\ref {sec:com-lo} and~\ref {sec:com-up}\\
\hline
\end{tabular}
\caption {Results in this paper; open spaces are directly implied by other bounds.}
\end {table}

In this paper we show that the maximum number of convex polygons in a mesh
can be as large as $\Omega(1.5028^n)$, and we give an upper bound of $O(1.62^n)$.
If the mesh is $\delta$-fat,
there can be $\Omega(n^{\frac12\lfloor\frac{2\pi}\delta\rfloor})$ convex polygons,
and we show an upper bound of $O(n^{\lceil\frac\pi\delta\rceil})$.
It is interesting (albeit not surprising) to see a case where fatness 
reduces the combinatorial complexity from exponential to polynomial.
For carrots, the compactness assumption removes the dependency on the
fatness parameter from the exponent of $n$.
Table~\ref {tab:results} summarizes our results. Note that $\rho $ and $\gamma$ do
not show up; the bounds hold for any constant values of $\rho$ and $\gamma$.

\begin {remark}
In our upper bound arguments, we do not require points to be in general position, and we count all sets of triangles whose union is convex, even if it has collinear edges of $M$ on its boundary.
In the lower bound constructions, we also use collinear points.
This simplifies the exposition, and it is natural to count all convex sets.
If one does want to exclude ``degenerate potatoes'', then the constructed meshes can be perturbed to be in general position, unless when $2\pi/\delta$ is precisely an integer:
in this setting, the number of fat potatoes in a fat mesh becomes
$\Omega(n^{\frac12\lceil\frac{2\pi}\delta\rceil-1})$,
and the number of fat carrots in a fat mesh becomes 
$\Omega(n^{\lceil\frac{2\pi}{3\delta}\rceil-1})$.

\end {remark}

\section {Preliminaries}

A \emph {mesh} is a plane straight-line graph with a finite set of vertices, such that 
all bounded faces are triangles, the interiors of all triangles
are disjoint and the intersection of any pair of triangles is either a vertex or a shared edge. 
  We also denote the set of vertices of a graph $G$ by $V(G)$ and the set of edges by $E(G)$, and say the \emph {size} of $G$ is $n = |V(G)|$.
  We say a mesh $M$ is \emph {maximal} if its triangles completely cover the convex hull of its vertices.\footnote {A maximal mesh is also called a \emph {triangulation}.}
A polygon $P$ is said to \emph {respect} a graph $G$ if all of its 
edges belong to $G$.

  We assume a mesh $M$ is given. 
  We call $M$ \emph {$\delta$-fat}, for some $\delta \in (0, \frac23\pi]$, if every angle of every triangle of $M$ is at least $\delta$.

Let $S = [0, 2\pi)$. We define cyclic addition and subtraction 
$(+, -) : S \times \R \to S$ in the usual way, modulo $2\pi$. 
We call the elements of $S$ \emph {directions} and implicitly 
associate an element $s \in S$ with the vector $(\sin s, \cos s)$.

\section {Potatoes in general meshes} \label {sec:gen}
\subsection {Lower bound} \label {sec:gen-lo}

Let $Q$ be a set of $m$ points evenly spaced on the upper half of a circle.
Assume $m=2^k+1$ for some integer $k$, and
let the points be $v_0,\ldots,v_{m-1}$, clockwise.
Let $M$ consist of the convex hull edges,
then connect $v_0$ and $v_{m-1}$ to $v_{(m-1)/2}$, and recursively
triangulate the subpolygons by always connecting the furthest pair to
the midpoint. The dual of the mesh is a perfectly balanced
binary tree.
Figure~\ref {fig:lo-gen-construction} illustrates the construction.

\eenplaatje {lo-gen-construction}
{ A set $Q$ of $n$ points on a half-circle, triangulated such that the dual tree is a balanced binary tree.
}

Let $N(k)$ be the number of different convex paths in $M$ from $v_0$ to $v_{m-1}$.
Then we have 
  \begin {equation*}
    N(k) = 1 + (N(k-1))^2,\;\;  N(1) = 2 
  \end {equation*}
because we can combine every path from $v_0$ to $v_{(m-1)/2}$ with
every path from $v_{(m-1)/2}$ to $v_{m-1}$, and the extra path is $\overline {v_0,v_{m-1}}$
itself. Using this recurrence we can relate the number $m$ of vertices used
to the number $P(m)$ of convex paths obtained: $P(3)=2$; $P(5)=5$; $P(9)=26$; 
$P(17)=677$. 

Now we place $n$ points evenly spaced on the upper half of a circle.
We triangulate $v_0,\ldots,v_{16}$ as above, and also $v_{16}\ldots,v_{32}$,
and so on. We can make $n/16$ groups of $17$ points
where the first and last point of each group are the same.
Each group is triangulated to give $677$ convex paths; the rest
is triangulated arbitrarily.
In total we get $677^{n/16}=\Omega(1.5028^n)$ convex paths from $v_0$ to $v_{n-1}$.
We omit the one from $v_0$ directly to $v_{n-1}$, and use this edge to complete
every convex path to a convex polygon instead. The number of convex
polygons is $\Omega(1.5028^n)$.\footnote 
{We can, of course, make larger groups of vertices to slightly improve the lower bound, but this does not appear to affect the given $4$ significant digits.}

  \begin {theorem}
    There exists a mesh $M$ with $n$ vertices such that the number of convex polygons that respect $M$ is $\Omega(1.5028^n)$.
    This is true even if $M$ is the Delaunay triangulation of its vertices.
  \end {theorem}

\subsection {Upper bound} \label {sec:gen-up}

Let $M$ be any mesh with $n$ vertices.
First, fix a point $p$ inside some triangle of $M$, not collinear with any pair of vertices of $M$. 
We will count only the polygons that contain $p$ for now.
    
    For every vertex $v$ of $T$, let $e_v$ be the edge that $v$ projects onto from $p$. (Assume general position of $p$ w.r.t.\ the vertices.) 
    Let $G$ be the graph obtained from $T$ by removing all such edges $e_v | v \in V(T)$.
    Figure~\ref {fig:up-gen-reduction+up-gen-g} shows an example.
 We turn $G$ into a directed graph by orienting every edge such that $p$ lies to the left
of its supporting line. We are interested in the number of simple cycles that respect $G$.
    Note that $G$ has exactly $2n - 3$ edges, since every vertex not on the convex hull causes one edge to disappear.

\tweeplaatjes {up-gen-reduction} {up-gen-g} {(a) We project each interior vertex of $M$ from $p$ onto the next edge. (b) The graph $G$ obtained by removing the marked edges.}

    We say a path in $G$ is \emph {monotone} around $p$ if the direction from $p$ to its vertices increases monotonically.

\begin{lemma}
The number of convex polygons in $M$ that have $p$ in their interior
is bounded from above by the number of simple cycles in $G$.
\end{lemma}
    \begin {proof}
      With each convex polygon, we associate a cycle by replacing any edges $e_v$ that were removed by the two edges via $v$, recursively.
      This results in a proper cycle because the convex polygon was already a monotone path around $p$, and this property is maintained.
      Each convex polygon results in a different cycle because the angle from the vertices of $e_v$ via $v$ is always concave.
    \end {proof}

  \begin {observation} \label {cla:star}
    The complement of the outer face of $G$ is star-shaped with $p$ in its kernel. 
  \end {observation}

  \begin {observation} \label {cla:deg1}
    Let $e$ be an edge on the outer face of $G$ from $u$ to $v$. Then $u$ has outdegree $1$, or $v$ has indegree $1$ (or both).
  \end {observation}

  \begin {proof}
    By Observation~\ref {cla:star}, there are no outgoing edges from $u$ to the right of $e$, and no incoming edges to $v$ to the right of $e$. So, if the outdegree of $u$ is more than one and the indegree of $v$ is more than one, there must have been a vertex in $T$ inside the triangle formed by $p$ and $e$. But then, this vertex would have been projected from $p$ onto $e$, and $e$ would not have been in $G$.
  \end {proof}

  If $F \subset E(G)$ is a subset of the edges of $G$, we also consider the subproblem of counting all simple cycles in $G$ that use all edges in $F$, the \emph {fixed} edges.
  For a tuple $(M, G, F)$, we define the \emph {potential} $\rho$ to be the number of vertices of $M$ (or $G$) minus the number of edges in $F$, i.e., $\rho(M,G,F) = |V(G)| - |F|$.
  Clearly, the potential of a subproblem is an upper bound on the number of edges that can still be used in any simple cycle.
  
  We will now show that the number of cycles in a subproblem can be expressed in terms of subproblems of smaller potential. Let $Q(k)$ be the maximum number of simple cycles in any subproblem with potential $k$.
  
  \begin {lemma}
    The function $Q(\cdot)$ satisfies
    \begin {equation*}
    Q(k) \leq Q(k-1) + Q(k-2),\;\; 
    Q(0) = Q(1) = 1\,.
    \end {equation*}
  \end {lemma}
  
  \tweeplaatjes {up-gen-case1} {up-gen-case2} {Two cases for $e$.}
  
  \begin {proof}
    Let $(M,G,F)$ be a subproblem and let $k = \rho(M,G,F)$. If $k = 1$ then $|F| = |V(G)| - 1$, so the number of fixed edges on the cycle is one less than the number of vertices available. Therefore the last edge is also fixed, if any cycle is possible. If $k=0$, all edges are fixed.
    
    For the general case, suppose all edges on the outer face of $G$ are fixed. Then there is only one possible cycle. 
If any vertex on the outer face has degree $2$ and only one incident edge fixed, we fix
the other incident edge too. Suppose there is at least one edge, $e = \overline {uv}$, on the outer face that is not fixed. By Observation~\ref {cla:deg1} one of its neighbours must have degree $1$ towards $e$. Assume without loss of generality that this is $v$. We distinguish two cases.

(i) The degree of $v$ is $2$, as illustrated in Figure~\ref {fig:up-gen-case1}. 
Any cycle in $G$ either uses $v$ or does not use $v$.
If it does not use $v$ we have a subproblem of potential $k-1$.
If it uses $v$, it must also use its two incident edges, so we can include these edges in $F$ to obtain a subproblem of potential $k-2$. 
So, the potential $\rho(M,G,F) \leq Q(k-1) + Q(k-2)$.

(ii) The degree of $v$ is larger than $2$, as illustrated in Figure~\ref {fig:up-gen-case2}.
Any cycle in $G$ either uses $e$ or does not use $e$. If it uses $e$, we can add $e$ to $F$ to obtain a subproblem of potential $k-1$. If it does not use $e$, then consider $v$ and the edge $e' = \overline{vw}$ that leaves $v$ on the outer face. Since $v$ has indegree $1$ but total degree greater than $2$, it must have outdegree greater than $1$. Therefore, by Claim~\ref {cla:deg1}, $w$ must have indegree $1$. Therefore, also $w$ will not be used by any cycle in $G$ that does not use $e$, and we can remove $v$ and $w$ to obtain a smaller graph. Again, we also remove all incident edges; if any of them was fixed we have no solutions.
      We obtain a subproblem of potential $k-2$ in this case. Again, the potential $\rho(M,G,F) \leq Q(k-1) + Q(k-2)$.
  \end {proof}

    This expression grows at a rate of the root of $x^2 - x - 1 = 0$, which is approximately $1.618034$.
  
    Because every convex polygon must contain at least one triangle of $M$, we just place $p$ in each triangle and multiply the bound by $2n$. Since $1.62$ is a slight overestimate (by rounding) of the root,
    we can ignore the factor $2n$ in the bound.

  \begin {theorem}
    Any mesh $M$ with $n$ vertices has $O(1.62^n)$ convex polygons that respect $M$.
  \end {theorem}

\section {Potatoes in fat meshes}  \label {sec:fat}
\subsection {Lower bound} \label {sec:fat-lo}

  Let $k = \lfloor\frac{2\pi}\delta\rfloor$, and let $l = \sqrt{\frac{n}{2k}}$. 
  Let $Q$ be a regular $k$-gon, and for each edge $e$ of $Q$ consider the intersection point of the supporting lines of the neighbouring edges.
  Let $Q'$ be a scaled copy of $Q$ that goes through these points.
  Now, consider a sequence $Q = Q_1, Q_2, \ldots, Q_l = Q'$ of $l$ scaled copies of $Q$ such that the difference between the radii of consecutive copies is always equal.
  We extend the edges of each copy until they touch $Q'$.
  Figure~\ref {fig:lo-fat-essential} illustrates the construction so far.

\drieplaatjes {lo-fat-essential} {lo-fat-final} {lo-fat-perturbed}
{ (a) Essential part of the construction, allowing $l^k$ choices for a convex polygon.
  (b) Final mesh.
  (c) The construction can be perturbed so as not to have collinear vertices.
}

  \begin {claim}
  The graph constructed so far allows at least $l^k$ different convex polygons.
  \end {claim}

  We now add vertices and edges to turn the construction into a $\delta$-fat mesh. 
  We use $l-1 \choose 2$ more vertices per sector, placing $l - i$ vertices on each edge of $Q_i$ to ensure that all angles are bounded by $\delta$.
  We need $O(lk)$ vertices to triangulate the interior using some adaptive mesh generation method.
  The final mesh can be seen in Figure~\ref {fig:lo-fat-final}. 

 The construction uses $\frac32kl^2 + O(kl)$ vertices, and since we have 
$l = \sqrt {\frac{n}{2k}}$, there are 
$\frac32kl^2 + O(kl)=\frac32k\frac{n}{2k}+O(k\sqrt {\frac{n}{2k}})=\frac34n +O(\sqrt{nk})\leq n$ vertices in total.

  Observe that the triangles of the outer ring are Delaunay triangles. The inner part can also be triangulated with Delaunay triangles, since the Delaunay triangulation maximises the smallest angle of any triangle.

  \begin {theorem}
    There exists a $\delta$-fat mesh $M$ of size $n$ such that the number of convex polygons that respect $M$ is $\Omega (n^{\frac12\lfloor\frac{2\pi}\delta\rfloor})$.
    This is true even if $M$ is required to be the Delaunay triangulation of its vertices.
  \end {theorem}

  \begin {remark}
    The construction described is highly degenerate. If $2\pi$ is not an integer multiple of $\delta$, then this is not essential, since
    the potatoes we count all contain a common point (i.e., they will not use both sides of
    any straight ($180^\circ$) angle). Therefore, we can perturb the resulting mesh slightly,
    and the same asymptotic bound applies to meshes in general position.
    Figure~\ref {fig:lo-fat-perturbed} illustrates the resulting shape.
  \end {remark}

\subsection {Upper bound} \label {sec:fat-up}

  In this section, we assume that $M$ is $\delta$-fat.
We consider paths of edges in $M$, where all edges 
of the path have roughly the same direction.

\begin {lemma}
\label{lem:onepath}
  Let $u, v \in V(M)$ be two vertices, and let $c, d \in S$ be two directions such that $d- c \leq 2\delta$.
  Then there is at most one convex path in $M$ from $u$ to $v$ that uses only directions in $[c, d)$.
\end {lemma}

  \eenplaatje {up-fat-extreme}
  { Two vertices $u$ and $v$ that need to be extreme in two directions that differ by at most $2 \delta$ (indicated by red and blue) define a unique potential convex chain since there can be at most one edge in each sector. 
    In this figure, the two paths do intersect, but do not form a convex chain.
  }  

\begin {proof}
  Let $m = c + \frac12 (d - c)$ be the direction bisecting $c$ and $d$. 
  Because $M$ is $\delta$-fat, for any vertex in $V(M)$ there is at most one incident edge with outgoing direction in $[c,m)$, and also at most one with direction in $[m,d)$.
  Because the path needs to be convex, it must first use only edges from $[c, m)$ and then switch to only edges from $[m,d)$.
  We can follow the unique path of edges with direction in $[c, m)$ from $u$ and the unique path of edges with direction in $[m + \pi, d + \pi)$ from $v$.
  If these paths intersect the concatenation may be a unique convex path from $u$ to $v$ as desired (clearly, the path is not guaranteed to be convex).
  Figure~\ref {fig:up-fat-extreme} illustrates this.
\end {proof}

  Given a convex polygon $P$ that respects $M$, a vertex $v$ of $P$ is \emph {extreme} in direction $s \in S$ if there are no other vertices of $P$ further in that direction, i.e. if $P$ lies to the left of the line through $v$ with direction $s + \frac12 \pi$.

  Let $\Gamma_\delta = \{0, 2\delta, 4\delta, \ldots, 2\pi\}$ be a set of directions. 
  As an easy corollary of Lemma~\ref{lem:onepath}, the vertices of a convex polygon $P$ respecting $M$ that are extreme in the directions of $\Gamma_\delta$ uniquely define $P$.
  There are at most $n$ choices for each extreme vertex, so the number of convex polygons is at most $n^{|\Gamma_\delta|}$.
  Substituting $|\Gamma_\delta| = \lceil\frac\pi\delta\rceil$ we obtain the following theorem.

\begin {theorem} \label {thm:fat-up}
  Any $\delta$-fat mesh $M$ of size $n$ has at most $O (n^{\lceil\frac\pi\delta\rceil})$ convex polygons that respect $M$.
\end {theorem}

\section {Carrots in fat meshes} \label {sec:car}

Recall that carrots are potatoes that have no interior vertices from the mesh.
So we expect fewer carrots than potatoes. However, our lower bound construction 
for general meshes only has potatoes that are also carrots. In this section we
therefore consider carrots in fat meshes.

\subsection {Lower bound}\label {sec:car-lo}

\tweeplaatjes [scale=1] {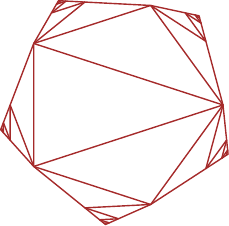} {lo-car-tower}
 { (a) An example of a $\delta$-fat mesh that has $\Omega(n^k)$ carrots.
   (b) A tower of $\delta$-$\delta$-$(\pi-2\delta)$ triangles.
 }

Let $k = \lfloor2\pi/3\delta\rfloor$, and consider a regular $k$-gon $Q$. On each edge of $Q$, we place a triangle with angles $\delta$, $2\delta$, and $\pi-3\delta$. Then, we subdivide each such triangle into $\frac{n-k}k$ smaller triangles with angles $\delta$, $\delta$, and $\pi-2\delta$, as illustrated in Figure~\ref {fig:lo-car-tower}.
Finally, we triangulate the internal region of $Q$ in any way we want, giving a mesh $M$.

\begin {lemma}
  $M$ is convex, $\delta$-fat, and contains $\Omega (n^{\lfloor\frac{2\pi}{3\delta}\rfloor})$ carrots.
\end {lemma}

\begin {proof}
  $M$ is convex because $\delta+2\delta \leq \frac{2\pi}k$.
  Every angle in the triangles outside $Q$ is at least $\delta$, and the angles in the interior of $Q$ are multiples of $\frac {\pi}k > \delta$.
  Therefore, every connected subset of $M$ is a carrot.
The dual tree $T$ of $M$ has a central component consisting of $k$ vertices, and then $k$ paths of length $\frac nk-1$. Hence, the number of subtrees of $T$ is at least $(\frac nk-1)^{k}$, which is $\Omega (n^{\lfloor\frac{2\pi}{3\delta}\rfloor})$.
\end {proof}

We conclude:

  \begin {theorem}
    There exists a $\delta$-fat mesh $M$ of size $n$ such that the number of convex outerplanar polygons that respect $M$ is $\Omega (n^{\lfloor\frac{2\pi}{3\delta}\rfloor})$.
  \end {theorem}
  
  \begin {remark}
    As in the previous section, if $2\pi$ is not an integer multiple of $\delta$, the constructed mesh can be perturbed to be in general position and the same bound holds.
  \end {remark}
  
\subsection {Upper bound}\label {sec:car-up}

In this section, we will show that given any $\delta$-fat mesh $M$, the number of carrots that respect $M$ can be at most $O(n^{\lfloor\frac{2\pi}{3\delta}\rfloor})$.

Consider any carrot. We inspect the dual tree $T$ of the carrot, and make some observations. Each node of $T$ is either a \emph {branch node} (if it has degree $3$), a \emph {path node} (if it has degree $2$), or a \emph {leaf} (if it has degree $1$). Path nodes have one edge on the boundary of the carrot, and leaves have two edges on the boundary of the carrot.
Figure~\ref {fig:up-car-tree} shows an example.

\tweeplaatjes [scale=1.5] {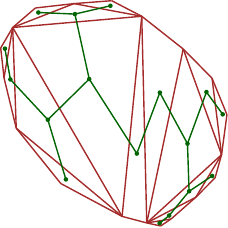} {up-car-core}
{ (a) A carrot and its dual tree.
  (b) The skeleton of the dual tree is the spanning tree of all vertices of degree $2$.
}
\begin {observation} \label {obs:leaf2delta}
  Let $v$ be a leaf node of $T$. The turning angle between the two external edges of $v$ is at least $2\delta$.
\end {observation}
\begin {proof}
  The triangle corresponding to $v$ is $\delta$-fat, so all three angles are at least $\delta$. Therefore, the angles are at most $\pi-2\delta$, and the turning angles are at least $2\delta$.
  Figure~\ref {fig:up-car-2delta} illustrates this.
\end {proof}

\begin {observation} \label {obs:edge3delta}
  Let $v$ be a leaf node of $T$ and $u$ a path node adjacent to $v$.
  The turning angle between the external edge of $u$ and the furthest external edge of $v$ is at least $3\delta$.
\end {observation}
\begin {proof}
  Consider the quadrilateral formed by the two triangles of $u$ and $v$.
  The edge in $M$ separating $u$ from the rest of $T$ has two $\delta$-fat 
triangles incident to one of its endpoints, and one to its other endpoint.
This means that the turning angle between the edges in the observation is $\geq 3\delta$
  (Figure~\ref {fig:up-car-3delta}).
\end {proof}

\tweeplaatjes {up-car-2delta} {up-car-3delta}
{(a) Every leaf causes a turning angle of $2\delta$.
 (b) Every leaf who is an only child causes a turning angle of $3\delta$.
}

By Observation~\ref {obs:leaf2delta}, the number of leaves in a carrot is bounded by $\lfloor\frac\pi\delta\rfloor$, and therefore, also the number of branch nodes is bounded by $\lfloor\frac\pi\delta\rfloor-2$. However, the number of path nodes can be unbounded.
Consider subtree $S$ of $T$ that is the spanning tree of all the path nodes. We call $S$ the \emph {skeleton} of the carrot. Figure~\ref {fig:up-car-core} shows an example.
By Observation~\ref {obs:edge3delta}, the number of leaves of $S$ is bounded by $\lfloor\frac{2\pi}{3\delta}\rfloor$.

We will charge the carrot to the set of leaves of $S$, and we will argue that every set of $\lfloor\frac{2\pi}{3\delta}\rfloor$ triangles in $M$ is charged only constantly often (for constant $\delta$).

\begin {observation} \label {obs:skeleton}
  Let $\Delta$ be any set of triangles of $M$.
  If there exists a carrot that contains all triangles in $\Delta$, then there is a unique smallest such carrot.
\end {observation}

\begin {lemma}
  Let $\Delta$ be any set of triangles in $M$.
  The number of carrots that charge $\Delta$ is at most $2^{\lfloor\frac{2\pi}\delta\rfloor}$.
\end {lemma}
\begin {proof}
  Consider the tree $S$ that is the dual of the unique smallest carrot that contains $\Delta$, as per Observation~\ref {obs:skeleton}. Any carrot that charges $\Delta$ has $S$ as its skeleton.
  
  First, we argue that the set of path nodes in any carrot that charges $\Delta$ is a subset of $S$.
  Indeed, if there was any path node in $T$ outside $S$, then there would be at least one leaf component of $T$ that is disconnected from $S$, and there would be an edge outside $\Delta$ that gets charged by the carrot of $T$.
  Therefore, only branch nodes and leaves can still be added to $S$ to obtain a carrot that charges~$\Delta$.
    
  Then, we argue that there are at most $2^{\lfloor\frac{2\pi}\delta\rfloor}$ other nodes that can be part of a carrot that charges $\Delta$.
  We can augment $S$ by adding on components consisting of only $k$ leaves and $k-1$ branch nodes.
  By Observation~\ref {obs:leaf2delta}, each such component consumes a turning angle of $2k\delta$. Therefore, they can only be added on edges of $S$ which have a cap angle of at least $2k\delta$. Therefore, there can be at most $2\pi/\delta$ potential leaves, leading to $2\lfloor\frac{2\pi}\delta\rfloor$ choices.\footnote {Not all potential leaves can be chosen independently, but we ignore this issue since the factor is dominated by the dependency on $n$ anyway.}
\end {proof}

We conclude:

\begin {theorem}
  Any $\delta$-fat mesh $M$ of size $n$ has at most $O(n^{\lfloor\frac{2\pi}{3\delta}\rfloor})$ convex outerplanar polygons that respect $M$.
\end {theorem}

\section {Carrots in compact fat meshes} \label {sec:com}

When the mesh is not only fat, but the edge length ratio
is also bounded by a constant, we can prove different bounds.
We call such meshes compact fat.
The combination of fat triangles and an edge length ratio bound
implies that all triangles have the same area, up to a constant
factor. For simplicity we assume that all edges have constant length
and therefore all triangles have constant area. The following
lemma is easy to show (see also Moet et al.~\cite{mks}):

\begin{lemma}
Given a compact fat mesh and a line segment $s$ of length
$d$, the number of triangles intersecting $s$ is $O(d)$.
\end{lemma}

\subsection {Lower bound}\label {sec:com-lo}

We distinguish fat carrots and general carrots in compact fat meshes.
The trivial lower bound for fat carrots is $\Omega(n)$, for example
each separate triangle of the mesh is a fat carrot.

The simple lower bound for general carrots is quadratic:
we place the $n$ points on a $2\times n/2$ grid and triangulate the 
row of squares. Clearly there are $\Omega(n^2)$ carrots.

\subsection {Upper bound}\label {sec:com-up}

Let $C$ be a carrot and let $pq$ be its diameter. Then all
triangles properly intersecting $pq$ are part of $C$ by convexity.
All vertices of these triangles lie on the boundary of $C$, since
carrots do not have interior vertices. 
Since these vertices are at constant distance from each other, a longer
diameter implies that $C$ becomes less fat.

\begin {observation}
  For fat carrots, the diameter has constant length. 
\end {observation}

There are $O(1)$ pairs of vertices in a
compact fat mesh whose distance is a constant. For each such pair,
all fat carrots that have this pair as the diameter lie inside a
region of constant area (Figure~\ref{fig:compactfatfatcarrots}), 
and hence contains $O(1)$ triangles of the mesh by a packing argument. 
A fat carrot is some subset of these triangles. Hence, the total number 
of fat carrots in a compact fat mesh is $O(n)$.

\tweeplaatjes {compactfatfatcarrots} {basecarrotouteredge}
{(a) The region where triangles can be for a carrot with diameter $pq$.
 (b) A base carrot and extension possibilities on an outer edge $e$.}
  
\begin {theorem}
Any compact fat mesh $M$ of size $n$ has at most $O(n)$ convex fat outerplanar polygons that respect $M$.
\end {theorem}

We next prove that there are $O(n^2)$ carrots in compact fat
meshes. We show that for any
pair of vertices there can be at most constantly many carrots that 
have this pair as their diameter.

Let $p,q$ be two vertices of the mesh, and consider subsets
of triangles that have $pq$ as their diameter and that are a carrot.
The triangles intersecting the diagonal itself is a sequence of
triangles whose union must be convex, otherwise no carrot exists with
$pq$ as the diameter. So this sequence is a carrot called the 
\emph{base carrot} of $pq$. There can be other carrots with $pq$ as the
diameter. These carrots have (sequences of) triangles attached to the
edges of the base carrot in some restricted manner. We will show
that only constantly many triangles can be part of a carrot with
a given base carrot. This was already shown above when the diameter
has constant length, but we can prove this in general.

Let $\ell_p$ be the line through $p$ perpendicular to $pq$ and let
$\ell_q$ be the line through $q$ perpendicular to $pq$. We can assume
that the parallel lines $\ell_p$ and $\ell_q$ are at distance more than
$3$ times the longest edge in the mesh (otherwise we already
know that there are only $O(n)$ carrots with $pq$ as the diagonal).
Then there are at least four edges from $p$ to $q$ on the base carrot,
and vice versa. 

For any outer edge $e=\overline {vw}$ of the base carrot we examine whether it
can have a triangle attached and still be a carrot, and if so, how 
many triangles can be attached to it. 

\begin{lemma}
At most $\lfloor 2\pi/\delta\rfloor$ edges can have a triangle attached.
\end{lemma}
\begin {proof}
By Observation~\ref{obs:leaf2delta}, any leaf of a carrot makes an angle of at least $2\delta$.
So, only edges of a base carrot that have a cap angle of at least $2\delta$ can be augmented by more triangles.
\end {proof}

Next we analyze how many triangles can be attached to an external edge $e=\overline {vw}$
of a base carrot with $pq$ as diameter.
Assume that neither $v$ nor $w$ is $p$ or $q$,
and that a clockwise traversal of the boundary of the base carrot encounters 
$p$, $\ldots$, $v$, $w$, $\ldots$, $q$ in this order, see Figure~\ref{fig:basecarrotouteredge}.
The lines through $p$ and $v$ and through $q$ and $w$ together with $e$
bound a triangular region whose area is at most constant; this is true
because $p$ and $q$ are sufficiently far apart with respect to the 
maximum edge length in $M$.
Since all triangles have constant area, at most
constantly many triangles can be attached to $e$ without violating
convexity of the carrot. Next, assume that $p=v$; all other cases are symmetric.
Now we use the line $\ell_p$ and the line through $q$ and $w$ and use the
same argument. We can use the line $\ell_p$ because if any triangle
attached to $e$ goes beyond $\ell_p$, then $pq$ cannot be the diameter.
In both cases, only constantly many triangles can be attached to $e$.

\begin {theorem}
Any compact fat mesh $M$ of size $n$ has at most $O(n^2)$ convex outerplanar polygons that respect $M$.
\end {theorem}

\begin {remark}
The question how many potatoes or fat potatoes can be in a compact fat mesh
is not discussed explicitly, but the lower bound construction of Section~\ref{sec:fat}
can be adapted to produce a compact fat mesh with asymptotically the same number
of vertices. All potatoes we counted are fat.
So there can be $\Omega(n^{\frac12\lfloor\frac{2\pi}\delta\rfloor})$ 
fat potatoes in a compact $\delta$-fat mesh. An upper bound of
$O (n^{\lceil\frac\pi\delta\rceil})$ follows directly from the more
general statement in Theorem~\ref{thm:fat-up}.
\end {remark}

\section {Discussion}

We investigated the maximum number of convex polygons that can be formed using
the edges of a given mesh only.
We provided a construction for a general mesh with $\Omega(1.5028^n)$ 
such convex polygons, and showed that there is an upper bound of $O(1.62^n)$. 
The upper and lower bounds for fat meshes match up to a constant factor (depending on $\delta$) if $\frac\pi\delta$ is an integer. 
If not, they differ by a factor $\sqrt n$ if the remainder is smaller than $\frac12$, or by a factor $n$ if the remainder is at least $\frac12$.

\section* {Acknowledgements}

We thank Stefan Langerman and John Iacono for detecting an error in an earlier
version of this paper.

 M.L. was supported by the Netherlands Organisation for Scientific Research (NWO) under grant 639.021.123.
 J.P. was supported by NSF Grant CCF-08-30272, by NSA, by OTKA under EUROGIGA
project GraDR 10-EuroGIGA-OP-003, and by Swiss National Science Foundation
Grant 200021-125287/1.

\small

\bibliographystyle {abbrv}
\bibliography {copo}

\end{document}